\documentclass[11pt]{article}

\usepackage[utf8]{inputenc}
\usepackage{amsmath, amssymb, amsthm, bm, dsfont, enumitem, mathtools, physics}
\usepackage[margin = 1in]{geometry}
\usepackage{graphicx}
\usepackage[colorlinks]{hyperref}
\usepackage[capitalize]{cleveref}
\usepackage{setspace}
\usepackage[none]{hyphenat}
\hypersetup{
  linkcolor = black,
  citecolor = black,
  urlcolor = blue
}

\usepackage[round]{natbib}
\theoremstyle{definition}

\newtheorem{assumption}{Assumption}

\newtheorem{theorem}{Theorem}
\newtheorem{corollary}{Corollary}

\newtheorem{remark}{Remark}

\title{Utility Maximization Under Endogenous Uncertainty}
\author{Ayush Gupta\thanks{Boston University, Department of Economics; guptaayu@bu.edu \\
I am indebted to Mark Kon, Alex Chan, and Jawwad Noor for many helpful discussions and suggestions.}}
\date{26 May 2026 \\
\href{https://www.ayush-gupta.com/Files/Endogenous Uncertainty.pdf}{Please click here for the latest version}}

\begin{document}

\setdisplayskipstretch{1}
\setlength{\abovedisplayskip}{-1.5\parskip}
\setlength{\belowdisplayskip}{\baselineskip - 1.5\parskip}

\maketitle
\thispagestyle{empty}

\begin{abstract}
This paper studies decision problems where the decision maker's choice of action affects the probability distribution of a payoff relevant random variable. We establish sufficient conditions for the existence of an expected utility maximizing action in such settings. The main requirement is a mild continuity condition on the family of possible distributions. We also show that this condition is a minimal requirement. Our result does not require common assumptions such as the monotone likelihood ratio property (MLRP) or the convexity of distribution functions condition (CDFC). It can therefore be used to prove the existence of an optimal action in many settings where existing results do not apply, including an important class of problems where the support of the random variable depends on the decision maker's choice and the density functions are not pointwise continuous.
\end{abstract}

\newpage

\section{Introduction}
\label{introduction}

The existence of an expected utility maximizing action is a minimal requirement for any model of decision-making under uncertainty. When the uncertainty facing the decision maker can be represented by a random variable with a fixed probability distribution, such an optimal action exists under mild continuity and compactness conditions. However, the relevant random variable need not have a fixed distribution. If the distribution depends on her choice of action, we say that the decision maker faces endogenous uncertainty. Such endogenous uncertainty is salient in many economic settings. For example, consider a firm that must choose how much to invest in R\&D. If the level of investment affects the probability of a technological breakthrough, the distribution of the relevant random variable (R\&D output) depends on the firm's action. Similarly, a worker's level of effort may affect the probability of her being promoted; an individual's choice of lifestyle may affect the distribution of her future health status.

In settings with endogenous uncertainty, existing results require strong assumptions such as the monotone likelihood ratio property (MLRP) and the convexity of distribution functions condition (CDFC). The existence of an optimal action therefore cannot be verified in a broad class of economic problems. In particular, when the support of the random variable depends on the decision maker's choice, existing results are often not applicable. This paper establishes a more general existence result that addresses these concerns. We show that it is sufficient for the family of probability distributions to satisfy a mild continuity condition. We also show that this is a minimal condition because an optimal action need not exist when the condition is violated.

The additional generality of our result allows us to prove existence in many important settings where existing results do not apply. For example, consider a firm that must choose from a continuum of R\&D projects. The relevant random variable is the chosen project's outcome. Investing in a more speculative project increases the probability of a technological breakthrough but it also increases the probability of complete failure. The distribution of the random variable thus depends on the firm's choice of project. CDFC does not hold in this environment because choosing a more speculative project shifts probability mass from intermediate outcomes to more extreme outcomes. Choosing a different project therefore changes the shape of the distribution, rather than simply shifting it toward one extreme. MLRP may be violated for a similar reason. The likelihood ratio may not be monotone because a technological breakthrough and complete failure are both more likely when the project is more speculative. Moreover, if the value of a technological breakthrough depends on the choice of project, the support of the random variable is not fixed. Then the probability densities may not be pointwise continuous. Nevertheless, our result can be used to prove the existence of an optimal choice of project.

The remainder of this section discusses some related literature; \cref{model} presents our model and assumptions; our main result and its proof are contained in \cref{results}; \cref{example} presents a more formal example to illustrate the significance of our result; and \cref{conclusion} concludes.

\subsection{Related Literature}

This paper builds on the classic literature that studies the conditions under which an expected utility maximizing action exists. For settings without endogenous uncertainty, \cite{jordan1977} provided a very general result for continuous and bounded utility functions. Even when the utility function is unbounded, \cite{kennan1981} showed that mild continuity and compactness conditions are sufficient for the existence of an optimal action. 

This paper is also related to the principal-agent literature. That is because, in standard models with moral hazard, the employment contract maps the worker's level of effort to a distribution over wages.\footnote{In a standard principal-agent model with moral hazard, the agent chooses a level of effort, $e$. Output, $y$ is a random variable but its distribution depends on $e$. The principal observes $y$ but not $e$. The agent's wage, $w$ depends on $y$. The agent's choice of $e$ thus affects the distribution of $w$. See \cite{holmstrom1979} for a formal description of the standard model.} The worker then solves a decision problem with endogenous uncertainty. Much of the existing literature obtains tractable solutions to this problem using the first-order approach. However, \cite{rogerson1985} showed that the first-order approach is valid only under additional assumptions such as MLRP and CDFC.

Endogenous uncertainty is salient in many other literatures as well. For example, \cite{hansen2001}, \cite{paciello2014}, and \cite{fajgelbaum2017} all study environments in which the decision maker's choice of action affects the uncertainty she faces.

\section{Notation and Assumptions}
\label{model}

$A$ is the set of actions available to the decision maker.

$X$ is the set of possible realizations of a random variable. $\mathcal{F}$ is the Borel $\sigma$-algebra on $X$.

The probability measure on $X$ depends on the decision maker's choice of action. The measure is given by $m_a$ when the decision maker chooses $a \in A$. Note that we do not require the support of the random variable to be the same for every choice of action.

$u \colon A \times X \to \mathbb{R}$ is the decision maker's utility function. The expected utility of $a$ is:

\begin{align*}
v(a) = \int_{X} \; u(a, x) \; dm_a (x)
\end{align*}

Notice that changing $a$ has two distinct effects on expected utility: a direct effect through the utility function and an indirect effect through the probability measure. This is the key difference between our setting and a setting without endogenous uncertainty -- the probability measure is fixed when there is no endogenous uncertainty.

Our existence result requires the following assumptions:

\begin{assumption}
\label{a1}

$A$ is a compact and first countable topological space.

\end{assumption}

\begin{assumption}
\label{a2}

$u$ is continuous on $A \times X$.

\end{assumption}

\begin{assumption}
\label{a3}

$\overline{u} (x) \equiv \sup_{a \in A} \abs{u(a, x)}$ is integrable with respect to all $m_a$.

\end{assumption}

\begin{assumption}
\label{a4}

Let $g \colon X \to \mathbb{R}$ be any continuous function that is integrable with respect to all $m_a$. Then $a_n \to a$ implies:

\begin{align}
\int_{X} \; g(x) \; dm_{a_n} (x) \; \to \; \int_{X} \; g(x) \; d m_a (x) \label{eq1}
\end{align}

\end{assumption}

\cref{a1} and \cref{a2} are ubiquitous in the literature. Note that only the set of actions is assumed to be compact; $X$ may be any arbitrary topological space. In particular, we allow the support of the random variable to be unbounded.

\cref{a3} is a generalization of Assumption A4' in \cite{kennan1981}. It ensures that expected utility is always finite, without requiring the utility function to be bounded.\footnote{The St. Petersburg paradox is an example of a decision problem ruled out by \cref{a3}.} Intuitively, \cref{a3} allows utility to be arbitrarily high for some realizations of the random variable, as long as the probability of such realizations is sufficiently low.

\cref{a4} formalizes our main restriction on the family of probability measures. Notice that this condition is stronger than weak convergence of measures. That is because weak convergence requires \eqref{eq1} to hold for continuous and bounded functions only. \cref{a4} requires \eqref{eq1} to hold for all continuous and integrable functions. The stronger condition is needed because we allow for unbounded utility functions.

If we limit ourselves to continuous and bounded utility functions, \cref{a4} can be relaxed and \cref{a3} can be dropped altogether.\footnote{The same can be achieved by limiting ourselves to compact $X$.} However, doing so excludes many economic problems of interest. For example, in a standard portfolio choice problem, the rate of return has unbounded support. As a result, the total return may be arbitrarily large and utility is typically not bounded.\footnote{In a standard portfolio choice problem, the investor can invest in a risky asset or a risk-free asset. The rate of return on the risk-free asset is a constant, $r_f$. The rate of return on the risky asset is a random variable, $r$. The total return is $y r + (w_0 - y) r_f$ when $w_0$ is the total investment and $y$ is the investment in the risky asset. The relevant random variable is the total return and its distribution depends on the investor's choice of $y$. See \cite{merton1969} for a formal description of the standard problem.}

\section{Results}
\label{results}

The decision maker wants to maximize her expected utility. Her optimization problem can be written as:

\begin{align*}
\max_a v(a) = \max_a \int_{X} \; u(a, x) \; d m_a (x)
\end{align*}

Our main result says that:

\newpage
\begin{theorem}

The decision maker's expected utility maximization problem has at least one solution under Assumptions 1 - 4.

\end{theorem}

\begin{proof}

$A$ is compact (\cref{a1}). So a maximizer exists if $v$ is upper semi-continuous.

$A$ is also first countable (\cref{a1}). So $v$ is upper semi-continuous if:

\begin{align}
\limsup_n \int_{X} \; u(a_n, x) \; dm_{a_n} (x) \leq \int_{X} \; u(a, x) \; dm_a (x) \label{eq2}
\end{align}

whenever $a_n \to a$. This proof will show that \eqref{eq2} holds under Assumptions 1 - 4. Note that the present approach is necessary because the utility function and the measure both depend on $n$. This dual dependence prevents us from using standard convergence theorems to move the limit inside the integral.

We start by defining:

\begin{align*}
u_m (x) &= \sup_{k \geq m} \; u(a_k, x)
\\
A_m &= \{ a \} \cup \{ a_k \colon k \geq m \}
\end{align*}

$A_m$ is a compact set because any open cover of $A_m$ will include an open set $O$ such that $a \in O$. $a_k \to a$ so all but finitely many elements of $A_m$ will lie in $O$. The remaining elements can then be covered by a finite subset of the open cover.

Then:

\begin{align*}
u_m (x) &= \sup_{a \in \{ a_k \colon k \geq m \}} \; u(a, x)
\\
&= \sup_{a \in A_m} \; u(a, x)
\\
&= \max_{a \in A_m} \; u(a, x)
\end{align*}

The second equality comes from the fact that $u$ is continuous (\cref{a2}). So adding a limit point does not change the supremum. The third equality comes from the fact that a continuous function on a compact set attains its supremum.

$u_m (x)$ is thus the maximum of a continuous function over a compact set. So Berge's Maximum Theorem says that $u_m (x)$ is continuous on $X$.

Next note that $u(a_n, x) \leq u_m (x)$ if $m \leq n$. So $\limsup_n u(a_n, x) \leq u_m (x)$ for every $m$. Thus for every $m$: 

\begin{align*}
\limsup_n \int_{X} \; u(a_n, x) \; dm_{a_n} (x) \leq \limsup_n \int_{X} \; u_m (x) \; dm_{a_n} (x)
\end{align*}

In particular:

\begin{align*}
\limsup_n \int_{X} \; u(a_n, x) \; dm_{a_n} (x) \leq \inf_m \; \limsup_n \int_{X} \; u_m (x) \; dm_{a_n} (x)
\end{align*}

Now recall that $\overline{u} (x) \equiv \sup_{a \in A} \abs{u(a, x)}$. So for every $m$:

\begin{align*}
\abs{u_m (x)} \leq \overline{u} (x)
\end{align*}

Therefore $u_m$ is integrable because $\overline{u}$ is integrable. Thus by \cref{a4}:

\begin{align*}
\inf_m \; \limsup_n \int_{X} \; u_m (x) \; dm_{a_n} (x) &= \inf_m \int_{X} \; u_m (x) \; dm_a (x)
\\
&= \lim_{m \to \infty} \int_{X} \; u_m (x) \; dm_a (x)
\end{align*}

The last equality comes from the fact that $u_m$ is decreasing in $m$. We can now apply the Dominated Convergence Theorem to get:

\begin{align*}
\lim_{m \to \infty} \int_{X} \; u_m (x) \; dm_a (x) &= \int_{X} \; \lim_{m \to \infty} \; u_m (x) \; dm_a (x)
\\
&= \int_{X} \; \lim_{m \to \infty} \sup_{k \geq m} \; u(a_k, x) \; dm_a (x)
\\
&= \int_{X} \; \limsup_{k \to \infty} \; u(a_k, x) \; dm_a (x)
\\
&= \int_{X} \; u(a, x) \; dm_a (x)
\end{align*}

We have thus shown that $v$ is upper semi-continuous. Therefore the decision maker's maximization problem has at least one solution.

\end{proof}

\begin{remark}

A symmetric argument can be used to show that $v$ is also lower semi-continuous. The expected utility function is therefore continuous under Assumptions 1 - 4.

\end{remark}

Theorem 1 shows that Assumptions 1 - 4 are sufficient to ensure the existence of an expected utility maximizing action. The key restriction is \cref{a4} because it controls how the probability measure responds to changes in $a$. The importance of \cref{a4} is immediate from the proof of Theorem 1.

\begin{corollary}

\cref{a4} is necessary to ensure that the expected utility function is upper semi-continuous for all decision problems satisfying Assumptions 1 - 3.

\end{corollary}

\begin{proof}

Fix $A$ such that \cref{a1} is satisfied.

Fix $X$ and a family of probability measures $\{ m_a \colon a \in A \}$ such that \cref{a4} is violated.

Then there exists a continuous function $g \colon X \to \mathbb{R}$ such that it is integrable with respect to all $m_a$ and for some $a_n \to a$:

\begin{align*}
\int_{X} \; g(x) \; dm_{a_n} (x) \not\to \int_{X} \; g(x) \; dm_a (x)
\end{align*}

Then either:

\begin{align}
\limsup_n \int_{X} \; g(x) \; dm_{a_n} (x) \; > \; \int_{X} \; g(x) \; dm_a (x) \label{eq3}
\end{align}

And / or:

\begin{align}
\liminf_n \int_{X} \; g(x) \; dm_{a_n} (x) \; < \; \int_{X} \; g(x) \; dm_a (x) \label{eq4}
\end{align}

Notice that \eqref{eq4} implies:

\begin{align*}
\limsup_n \int_{X} \; -g(x) \; dm_{a_n} (x) \; > \; \int_{X} \; -g(x) \; dm_a (x)
\end{align*}

Since $-g$ is also continuous and integrable, it is without loss of generality to assume \eqref{eq3}.

Now define $u(a, x) = g(x)$.

\cref{a2} is trivial because $g$ is continuous on $X$.

\cref{a3} is satisfied because $\sup_{a \in A} \abs{u(a, x)} = \abs{g(x)}$ and $g(x)$ is integrable.

Recall that:

\begin{align*}
v(a) = \int_{X} \; u(a, x) \; d m_a (x) = \int_{X} \; g(x) \; d m_a (x)
\end{align*}

So \eqref{eq3} implies that for some $a_n \to a$:

\begin{align*}
\limsup_n \; v(a_n) > v(a)
\end{align*}

We thus obtain a decision problem that satisfies Assumptions 1 - 3 but $v$ is not upper semi-continuous because \cref{a4} is violated.

\end{proof}

The proof shows that, whenever \cref{a4} is violated, there exists a utility function such that expected utility is not upper semi-continuous. It is clear that an optimal action need not exist without upper semi-continuity. \cref{a4} is therefore a minimal condition -- it cannot be weakened without imposing additional restrictions on the decision problem.

Note that \cref{a4} is a condition on the family of probability measures. It does not require the measures to admit density functions with respect to a common reference measure. Even when such densities exist, \cref{a4} is not equivalent to a pointwise continuity condition on those densities.

Formally, suppose $X$ is equipped with a $\sigma$-finite measure $\mu$ such that every $m_a$ is absolutely continuous with respect to $\mu$. Then every $m_a$ admits a density function $f_a$ with respect to $\mu$. This family of densities is pointwise continuous if for every $x \in X$:

\begin{align*}
f_{a_n} (x) \to f_a (x)
\end{align*}

whenever $a_n \to a$. We stress that pointwise continuity of the density functions is not necessary for \cref{a4}. \cref{example} provides an explicit example of a decision problem where the densities are not pointwise continuous but \cref{a4} is satisfied. Intuitively, pointwise continuity is most likely to fail when the support of the random variable depends on the choice of action.

\section{Example}
\label{example}

The decision maker is a physician.

$a \in [0, 1]$ is the dosage of an experimental cancer drug.

$x \in [-1, 1]$ is the patient's clinical outcome. $x = 1$ represents complete recovery while $x = -1$ represents death due to severe toxicity.

Let $U \left[ I \right]$ denote the uniform probability measure on the closed interval $I$. Then the probability measure on $X$ is:

\begin{align*}
m_a = U \left[ -\frac{1 + a}{2}, \frac{1 + a}{2} \right] 
\end{align*}

Note that the support of $x$ depends on the choice of $a$. Higher dosages allow for better clinical outcomes but also allow for worse outcomes.

The utility function is given by $u(a, x) = b(x) - c(a)$. $b(x)$ represents the value of the realized outcome and $c(a)$ captures the cost of the dosage. $c(a)$ can be interpreted as the disutility from the known side effects of the drug. Assume that $b$ and $c$ are both continuous functions.

Every probability measure in this example admits a density function with respect to the standard Lebesgue measure:

\begin{align*}
f_a (x) = \frac{1}{1 + a} \; \mathds{1} \left( \abs{x} \leq \frac{1 + a}{2} \right)
\end{align*}

However, this family of densities is not pointwise continuous.\footnote{The density functions are not unique, but there is no choice of densities that makes $a \mapsto f_a$ pointwise continuous.}

To see that, fix any $x \in \left[ -1, -\frac{1}{2} \right) \cup \left( \frac{1}{2}, 1 \right]$. Let $a_n \to a$ be any sequence such that $a_n < a$ for every $n$ and $\frac{1 + a}{2} = \abs{x}$. Then:

\begin{align*}
\abs{x} > \frac{1 + a_n}{2} \implies f_{a_n} (x) = 0 \not\to \frac{1}{1 + a} = f_a (x)
\end{align*}

The densities are thus discontinuous at $\frac{1 + a}{2} = \abs{x}$.

It is also clear that this example violates CDFC. That is because $F_a (x) = \frac{1}{2} + \frac{x}{1 + a}$ when $\abs{x} \leq \frac{1 + a}{2}$.

Differentiating twice with respect to $a$ gives $F''_a (x) = \frac{2x}{(1 + a)^3}$. Notice that $F''_a (x) < 0$ if $x < 0$.

The distribution function is thus concave, rather than convex, for all $x \in \left[ -\frac{1}{2}, 0 \right)$.

MLRP requires that:

\begin{align*}
f_{a'} (x_H) \; f_a (x_L) \geq f_{a'} (x_L) \; f_a (x_H)
\end{align*}

when $a' > a$ and $x_H > x_L$. Now fix:

\begin{align*}
-\frac{1 + a'}{2} < x_L < -\frac{1 + a}{2} < x_H < \frac{1 + a}{2} < \frac{1 + a'}{2}
\end{align*}

So $x_H$ is in the support of both $a$ and $a'$ but $x_L$ is only in the support of $a'$. Then:

\begin{align*}
f_{a'} (x_H) \; f_a (x_L) &= 0
\\
f_{a'} (x_L) \; f_a (x_H) &= \frac{1}{(1 + a)(1 + a')} > 0
\end{align*}

MLRP is thus violated in this example.

Nevertheless, Theorem 1 says that an expected utility maximizing action exists.

It is obvious that \cref{a1} and \cref{a2} are satisfied.

The utility function is continuous and $A$ and $X$ are both compact sets. So $\overline{u}$ is also bounded and hence integrable. \cref{a3} is therefore satisfied.

To see that \cref{a4} is satisfied, note that for any continuous and integrable function $g$:

\begin{align*}
\int_{X} \; g(x) \; d m_a (x) = \frac{1}{1 + a} \; \int_{-\frac{1 + a}{2}}^{\frac{1 + a}{2}} \; g(x) \; dx
\end{align*}

The Fundamental Theorem of Calculus says that $\int_{-\frac{1 + a}{2}}^{\frac{1 + a}{2}} \; g(x) \; dx$ is a continuous function of $a$ because $g$ is continuous.

$\frac{1}{1 + a}$ is also continuous so \cref{a4} is satisfied in this example.

Theorem 1 then ensures that an expected utility maximizing action exists.

\section{Conclusion}
\label{conclusion}

This paper studies decision problems where the distribution of a payoff relevant random variable depends on the decision maker's choice of action. Such endogenous uncertainty is common in many different fields of economics. However, existing results require strong assumptions that limit their applicability. This paper establishes a general existence result that can be used to prove the existence of an optimal action in a broad class of problems where existing results do not apply. Our result shows that it is sufficient for the family of probability distributions to satisfy a mild continuity condition. We also show that this condition is a minimal requirement because an optimal action need not exist when it is violated. The main technical challenge is integrating over all possible realizations of the random variable when the realized utility and the probability measure both depend on the choice of action. This dual dependence prevents straightforward applications of Fatou’s Lemma or conventional dominated convergence arguments. We overcome this challenge by using a tail envelope and show that the expected utility function is upper semi-continuous under our assumptions.

\bibliography{References}

\end{document}